\documentclass{article}
\usepackage{amsmath, amssymb}
\usepackage{amsthm}
\usepackage{bm}
\usepackage{comment}
\usepackage{graphicx}
\renewcommand{\P}{\mathbb{P}}
\newcommand{\E}{\mathbb{E}}
\usepackage{authblk}

\newtheorem{theorem}{Theorem}
\newtheorem{remark}{Remark}

\title{Information dissemination processes in directed social networks}

\author[1]{K. Avrachenkov\thanks{k.avrachenkov@sophia.inria.fr}}
\author[2]{K. De Turck\thanks{kdeturck@telin.ugent.be}}
\author[2]{D. Fiems\thanks{Dieter.Fiems@telin.UGent.be}}
\author[3,4]{B.J. Prabhu\thanks{balakrishna.prabhu@laas.fr}}
\affil[1]{INRIA Sophia Antipolis, 2004 route des Lucioles,  06902 Sophia Antipolis Cedex, France} 
\affil[2]{Department of Telecommunications and Information Processing, Ghent University, St. Pietersnieuwstraat 41, 9000 Gent, Belgium}
\affil[3]{CNRS, LAAS, 7 avenue du colonel Roche, F-31400 Toulouse, France}
\affil[4]{Univ de Toulouse, LAAS, F-31400 Toulouse, France}

\begin{document}
\maketitle

\begin{abstract}
Social networks can have asymmetric relationships. In the online social network
Twitter, a follower receives tweets from a followed person but the followed person is not obliged to
subscribe to the channel of the follower. Thus, it is natural to consider the dissemination of information
in directed networks. In this work we use the mean-field approach to derive differential equations that
describe the dissemination of information in a social network with asymmetric relationships. In particular,
our model reflects the impact of the degree distribution on the information propagation process. We further
show that for an important subclass of our model, the  differential equations can be solved analytically.
\end{abstract}

\section{Introduction}

We develop mathematical models for the dissemination of information on directed graphs and investigate the
influence of parameters such as the degree distribution on the dynamics of the dissemination process.
The directed graph model, as opposed to the undirected model, is better suited for networks like the
one of Twitter because of the asymmetric relationship that exists between different users. Specifically,
in the Twitter network, a user can choose to receive the tweets---in other words, become a follower---of one or more 
other users by subscribing to their accounts. Certain users, celebrities for example, have several millions of 
followers who follow their tweets. These users do not necessarily follow the tweets of all of their followers 
which results in an asymmetric relationship between users. This asymmetry is modelled by a directed graph in 
which an outgoing edge is drawn from a user to each of its followers. An edge in the opposite direction from 
the follower to the user need not always exist and is drawn only if the user subscribes to the channel of this follower.

A hashtag is a word or a phrase prefixed by \# and used in social networks as a keyword. The prefix facilitates 
the search for conversations related to the prefixed word or phrase. The typical life cycle of a hashtag  closely 
resembles an epidemic. In the first phase the interest in the hashtag grows as users
generate tweets containing this hashtag. These tweets are received by followers who then either retweet them
or generate new tweets with this hashtag. The number of users tweeting this hashtag (``infected'' users)
grows as a function of time during this phase. At a certain point in time, the interest reaches its zenith
and starts to wane as users move on and get interested in other events. The second phase begins at this point
in time as users stop tweeting this hashtag (or, ``recover''), and the number of infected users decreases.

In this work we use the mean-field approach to derive the differential equations which describe the process
of information dissemination. We obtain a couple of differential equations which describe the evolution
of the fractions of infected and recovered persons. As a model for the underlying network we take
the Configuration-type model for directed graphs \cite{CO-C13}. While epidemics have been widely studied on 
undirected graphs, there is a hardly any analysis of the epidemic-type processes on directed networks.
In \cite{PS01,M02}, the mean-field approach has been applied to the analysis of epidemics on an undirected
configuration-type graph model. In \cite{GMT05}, the effect of network topology has been analysed in
the case of undirected graphs. In particular, the authors of \cite{GMT05} applied their general results
to analyse the Erd\"{o}s-R\'enyi and preferential attachment random graph models.
An interesting approach combining a decomposition approach with two-state primitive Markov chain has been
proposed in \cite{vMOK09} for undirected networks with general topology. For an overview of various results
about epidemic processes on undirected networks we refer the interested reader to the books \cite{D07,BBV08,DM10}.
%\cite{D07,BBV08,DM10}.

\section{The mean-field model}
\label{sec:meanfield}
Consider a network of $N$ nodes structured according to the Configuration-type model for directed graphs \cite{CO-C13}. 
The in-degree and out-degree of the nodes are drawn from a distribution $f(k,l) = \P(K = k, L = l)$, defined on the 
bounded set $\mathcal{D} = \{(k,l) : 0 \leq k \leq \hat{K}, 0\leq l \leq \hat{L}, (k,l) \neq (0,0)\}$, where the first 
(resp. second) index corresponds to the in-degree (resp. out-degree) and $\hat{K}$ (resp. $\hat{L}$) is the maximal
in-degree (resp. out-degree). In the remainder, we always assume that $\E K = \E L$; in a network every outgoing link 
is an incoming link of some other node. A generic node with in-degree $k$ and out-degree $l$ shall be referred to as a $(k,l)$-node.

%Add a note on how a configuration type network is formed.  \cite{CO-C13}

Each node in the network can be in one of the three states : infected, recovered, or susceptible. An infected node infects its susceptible
neighbours after an exponentially distributed time with intensity $\lambda$. Note that, the infection is spread simultaneously along
all the outgoing edges and not just one edge at a time. The simultaneous dissemination along all outgoing edges models the spread of
tweets on Twitter. An infected node recovers after an exponentially distributed time of rate $\nu$, at which time it stops spreading
information in the network.

We shall be mainly interested in a
large-population model, that is when $N\to\infty$. This assumption simplifies considerably the analysis of the dissemination process 
while being realistic\footnote{Twitter has approximately $200$ million registered users
(source: Wikipedia).}.

Let $i_{k,l}(t)$ (resp. $r_{k,l}(t)$) denote the fraction of infected (resp. recovered) $(k,l)$ nodes
at time $t$. The following result describes the dynamics of these two quantities.

\begin{theorem}
		Let $i_{k,l}(0) > 0$ for some $(k,l)\in\mathcal{D}$. Then, $\forall (k,l) \in \mathcal{D}$,
\begin{equation}
	\frac{di_{k,l}(t)}{dt} = \lambda k (f(k,l)-i_{k,l}(t) - r_{k,l}(t))\frac{\sum_{k',l'}l'i_{k',l'}(t)}{\sum_{k',l'}l'f(k',l')} 
				- i_{k,l}(t)\nu,
\label{eqn:alt_master}
\end{equation}
and
\begin{equation}
		\frac{dr_{k,l}(t)}{dt} = i_{k,l}(t)\nu.
\label{eqn:alt_master_r}
\end{equation}
\end{theorem}

\begin{proof}[Sketch of proof]
		Let $I^{N}_{k,l}(t)$ (resp. $R^{N}_{k,l}(t)$) be the number of infected (resp. recovered) $(k,l)$ nodes in
		a network of $N$ nodes. Then in a small time interval $\Delta$,
\begin{align*}
		I^{(N)}_{k,l}(t+\Delta) &= I^{(N)}_{k,l}(t) + \mbox{number of $(k,l)$ nodes infected in time $\Delta$} \\
					&- \mbox{number of $(k,l)$ infected $(k,l)$ nodes that recover in time $\Delta$}.
\end{align*}
Since each infected node recovers after an exponentially distributed time of rate $\nu$, the number of $(k,l)$
infected nodes that recover in $\Delta$ will be approximately $I^{(N)}_{k,l}(t)\nu\Delta$. There will be additional
terms containing $\Delta^2$ which we neglect.

Let us compute the number of $(k,l)$ nodes that get infected in time $\Delta$. There are
$N^{(N)}_{k,l} - (I^{(N)}_{k,l}(t) + R^{(N)}_{k,l}(t))$ susceptible $(k,l)$ nodes. Assume that each $(k,l)$ node has a
probability $p_{k,l}$ to get infected in the interval $\Delta$. Then, expected number of infected $(k,l)$ nodes in time
$\Delta$ will be $(N^{(N)}_{k,l}-(I^{(N)}_{k,l} + R^{(N)}_{k,l}))p_{k,l}$.

Each $(k,l)$ node has $k$ incoming edges. Assuming that the edges are connected independently,
$p_{k,l} = (1 - (1-q_{k,l})^k)$, where $q_{k,l}$ is the probability that the infection is transmitted along one of the edges in
$\Delta$.
The number of $(k,l)$ nodes infected in $\Delta$ is thus
\[
		(N^{(N)}_{k,l}-(I^{(N)}_{k,l}+R^{(N)}_{k,;}))\cdot(1-(1-q_{k,l})^k),
\]
which, for $\Delta$ sufficiently small, can be approximated as
\[
		(N^{(N)}_{k,l}-(I^{(N)}_{k,l} + R^{(N)}_{k,l}))kq_{k,l}.
\]
Finally, we compute $q_{k,l}$. In an interval $\Delta$, each infected node transmits the infection with probability
$\lambda \Delta$. Thus, there are $\sum_{k','l}l'I^{(N)}_{k',l'}\lambda\Delta$ edges that are infected and that transmit the 
infection in $\Delta$. There are a total of $\sum_{k',l'}l'N^{(N)}_{k',l'}$. Assuming that an incoming edge is connected 
uniformly at random to an outgoing edge, we obtain
\[
		q_{k,l} = \frac{\sum_{k',l'}l'I^{(N)}_{k',l'}\lambda\Delta}{\sum_{k',l'}l'N^{(N)}_{k',l'}}.
\]

Consequently,
\[
		I^{(N)}_{k,l}(t+\Delta) - I^{(N)}_{k,l}(t) = (N^{(N)}_{k,l}-(I^{(N)}_{k,l} 
					+ R^{(N)}_{k,l}))k\frac{\sum_{k',l'}l'I^{(N)}_{k',l'}\lambda\Delta}{\sum_{k',l'}l'N^{(N)}_{k',l'}} 
					- I^{(N)}_{k,l}\nu\Delta.
\]
If the initial number of nodes is large, then we can divide the two sides of the above equation to obtain the following difference
equation in terms of the fraction of the infected and the recovered nodes:
\[
		i_{k,l}(t+\Delta) - i_{k,l}(t) = (f(k,l)-(i_{k,l}(t) 
				+ r_{k,l}(t)))k\frac{\sum_{k',l'}l'i_{k',l'}(t)\lambda\Delta}{\sum_{k',l'}l'f(k',l')}
				- i_{k,l}(t)\nu\Delta.
\]

To complete the picture, we take the limit $\Delta \to 0$, and obtain the differential equations \eqref{eqn:alt_master}
and \eqref{eqn:alt_master_r}.
\qed
\end{proof}

\begin{remark}
In the above equation $i_{k,l}$ is the fraction of the $(k,l)$ nodes that are infected. This fraction varies between $0$ and $f(k,l)$. If
instead, we want to look at the evolution of the fraction of infected nodes and recovered nodes conditioned on them being
$(k,l)$ nodes, then the corresponding differential equations for these fractions will be
\begin{align}
		\frac{di_{k,l}(t)}{dt} &= \lambda k (1-i_{k,l}(t) 
				- r_{k,l}(t))\frac{\sum_{k',l'}l'f(k'.l')i_{k',l'}(t)}{\sum_{k',l'}l'f(k',l')} 
				- i_{k,l}(t)\nu, \label{eqn:alt_master2} \\
		\frac{dr_{k,l}(t)}{dt} &= i_{k,l}(t)\nu.
\label{eqn:alt_master2_r}
\end{align}
\end{remark}

\section{Epidemics without recovery}
The solution of \eqref{eqn:alt_master} and \eqref{eqn:alt_master_r} can be computed numerically. In some specific case 
we can obtain explicit solutions to these equations. In particular, this is the case when there is no recovery: $\nu = 0$,  
or in the language of Twitter, they keep generating new tweets with the same hashtag. That is, a hashtag never gets out 
of mode. This can well represent the case for the topics or personalities that can sustain popularity over a long period of time.

Since there are no recovered nodes, $r_{k,l}(t) = 0, \, \forall t$, and \eqref{eqn:alt_master} takes the form
\begin{align}
	\frac{di_{k,l}(t)}{dt} &= \lambda k (f(k,l)-i_{k,l}(t))\frac{\sum_{k',l'}l'i_{k',l'}(t)}{\sum_{k',l'}l'f(k',l')}
	\label{eqn:master_nu0}.
\end{align}
The differential equation \eqref{eqn:master_nu0} can be solved in terms of a reference value of $(k,l)$, say 
$(k,l) = (1,1)$ by noting that
\[
		\frac{1}{k (f(k,l)-i_{k,l}(t))}\frac{di_{k,l}(t)}{dt} = \frac{1}{(f(1,1)-i_{1,1}(t))}\frac{di_{1,1}(t)}{dt},
\]
whence
\begin{align}
		f(k,l)-i_{k,l}(t) &= \frac{f(k,l)-i_{k,l}(0)}{(f(1,1) - i_{1,1}(0))^k}(f(1,1)-i_{1,1}(t))^k =: c_{k,l}(f(1,1)-i_{1,1}(t))^k.
\label{eqn:par_soln_alt_mas_nu0}
\end{align}

The fraction of infected nodes of degree $(1,1)$ can be obtained by substuting the value of $i_{k,l}(t)$ in
\eqref{eqn:master_nu0} and solving it:
\begin{equation}
		\frac{di_{1,1}(t)}{dt} = \lambda (f(1,1)-i_{1,1}(t)) \frac{\sum_{k',l'}l'
		\left(f(k',l') - c_{k', l'}(f(1,1)-i_{1,1}(t))^{k'}\right)}{\sum_{k',l'}l'f(k',l')}.
\label{eqn:soln_alt_master}
\end{equation}

\subsection*{Deterministic in-degree}
Assume that the in-degree $K$ is deterministic and is equal to $d$.
Then, equation \eqref{eqn:master_nu0} becomes
\[
		\frac{di_{d,l}(t)}{dt} = \lambda d (f(d,l)-i_{d,l}(t))\sum_{l'} \frac{l'i_{d,l'}(t)}{\sum_j jf(d,j)}.
\]
Since the expected in-degree and the expected out-degree coincide, $\sum_j jf(d,j) = d.$
Denote $\Theta(t) = \sum_j \frac{ji_{d,j}(t)}{d}$, and rewrite the above equation as:
 \begin{equation}
		 \frac{di_{d,l}}{dt} = \lambda d (f(d,l)-i_{d,l}(t))\Theta(t).
\label{eqn:mast_fix_in_theta}
\end{equation}
 Multiplying the above equation by $\frac{l}{d}$ and summing over all values $l$, we obtain the following equation for $\Theta$:
\[
\frac{d\Theta}{dt} = \lambda d(1-\Theta)\Theta,
\]
which upon integration yields:
\begin{equation}
		i_{d,l}(t) =  f(d,l) - c_1 e^{-\lambda d \int \Theta(t) dt}
		= f(d,l) - \frac{f(d,l)-i_{d,l}(0)}{1-\Theta(0)+\Theta(0) e^{-\lambda d \cdot t}}.
\label{eqn:mast_soln_fixin}
\end{equation}

\section{Numerical experiments}
In this section, we validate the mean-field model developed in Section \ref{sec:meanfield}.
In the numerical experiments, first the in-degree and out-degree sequences are generated according to the given
degree distributions. So as to have the same number of incoming stubs as outgoing stubs, the difference between the two is
added to the smaller quantity. A configuration-type graph is then created by matching an incoming stub with an outgoing
stub chosen uniformly at random. It was shown in \cite{CO-C13} that this procedure does indeed approximates closely the 
configuration model. The information diffusion process is then simulated on this graph.

For computing the solution of the system of differential equations \eqref{eqn:alt_master} and \eqref{eqn:alt_master_r} 
numerically, the empirical degree distributions from the graph generated previously are given as input.

The results of two such experiments with $20000$ nodes is shown in Figure \ref{fig:experiments}. The in-degree and the 
out-degree sequences were taken to be independent of each other. The out-degree sequence was drawn from a Uniform distribution 
in the set $\{1,20\}$ in the two simulations. For the figure on
the left, the in-degree distribution was taken to be deterministic with parameter $10$, and for the figure on the right it was
the Zipf law on $\{1,71\}$ and exponent $1.2$. In both experiments, $\lambda = 1$ and $\nu = 0.5$, and $5$ percent of all nodes
were assumed to be infected at time $0$.
\begin{figure}[ht]
		\begin{minipage}{0.45\linewidth}
				\centering\includegraphics[width=1.1\linewidth]{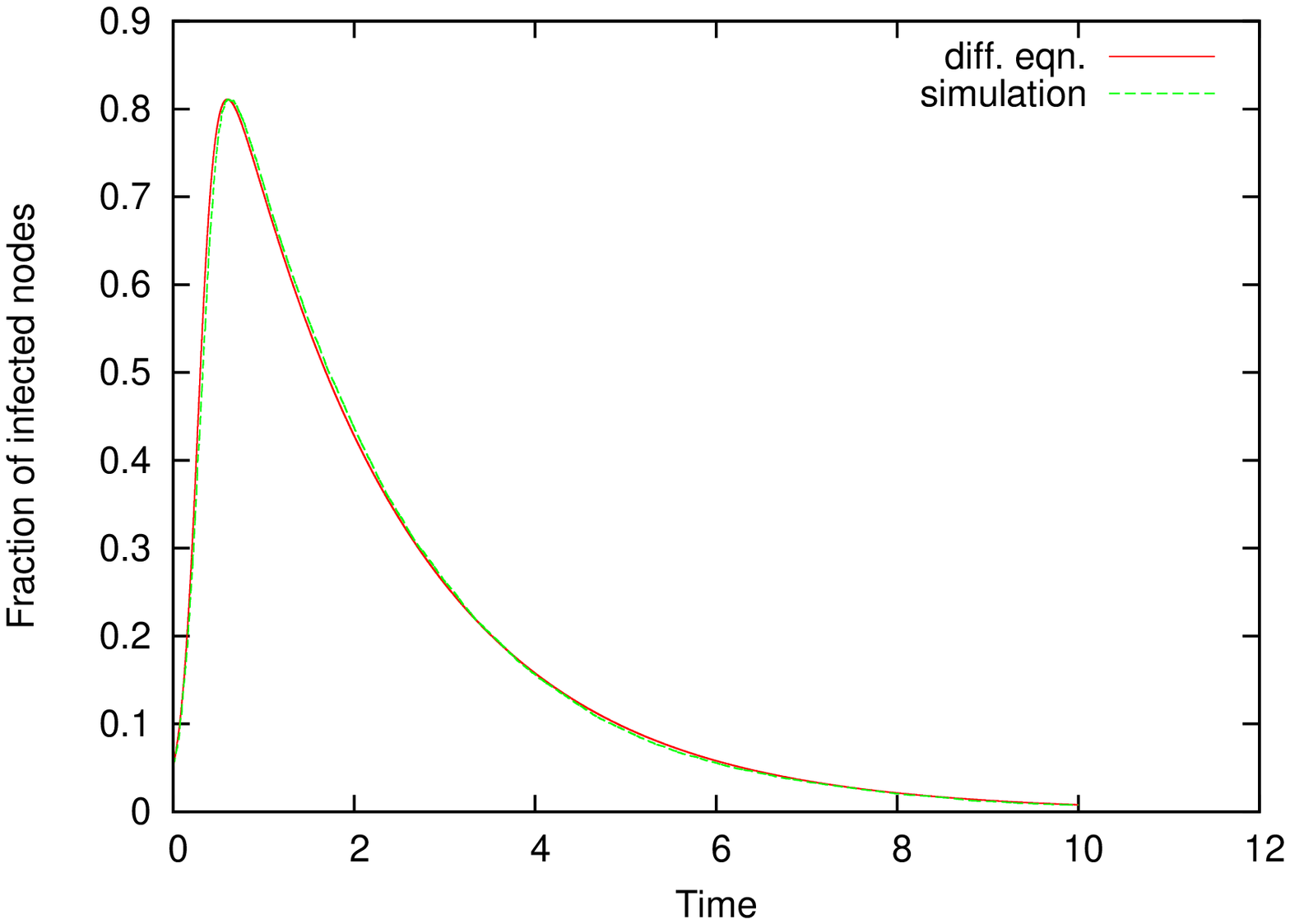}
		\end{minipage}
		\hspace{1cm}
		\begin{minipage}{0.45\linewidth}
				\centering\includegraphics[width=1.1\linewidth]{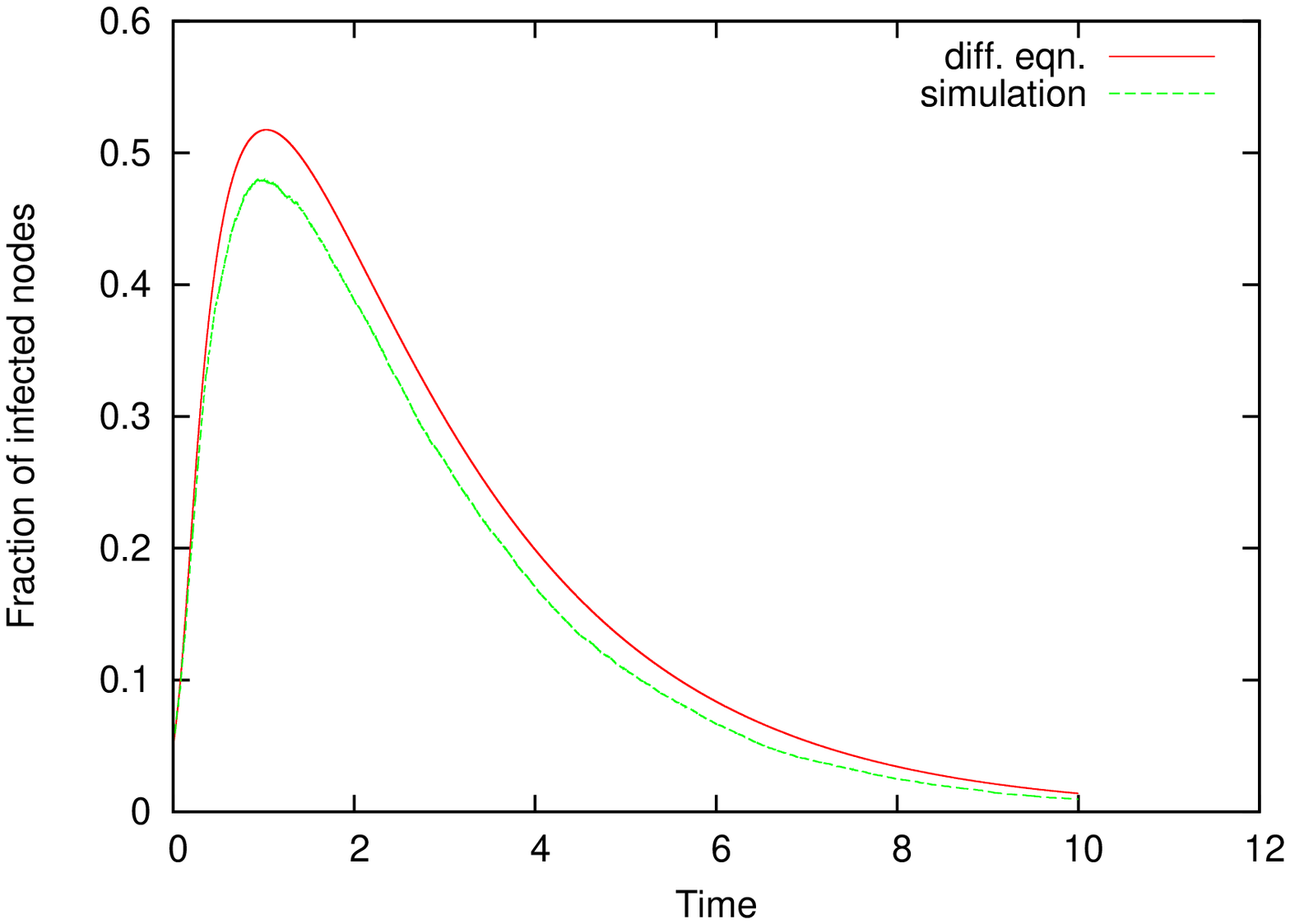}
		\end{minipage}
		\caption{Fraction of all nodes infected as a function of time for Deterministic in-degree distribution (left) and Zipf in-degree
				distribution (right).}
		\label{fig:experiments}
\end{figure}
\subsection*{Observations}
It is observed that the dissemination process is faster when the variance of the in-degree distribution is smaller. This
observation was reinforced by other experiments in which the in-degree was drawn from a Uniform distribution. In several
other experiments that we conducted, it was also observed that the out-degree distribution does not have 
any noticeable effect of the dynamics of the epidemics.

Our on-going work is oriented towards investigating the influence of the variance of the in-degree distribution
and giving a theoretical foundation to the above observations.
\section{Acknowledgments}
This work was partially supported by the Parternariat Hubert Curien PHC
TOURNESOL FR 2013 29053SF between France and the Flemish community
of Belgium, by the Inria Alcatel-Lucent Joint Lab ARC ``Network Science'', and by the European Commission within
the framework of the CONGAS project FP7-ICT-2011-8-317672.
\bibliographystyle{plain}
\bibliography{rr}
\end{document}